\documentclass[
final,
nomarks
]{dmtcs-episciences}

\usepackage[round]{natbib}

\usepackage{amssymb,mathtools} 
\usepackage{overpic,subcaption}
\usepackage{hyperref}
\usepackage[linesnumbered,ruled,vlined]{algorithm2e}
\usepackage{todonotes}\reversemarginpar

\newcommand{\mybreak} {\par\vspace{2mm}\noindent}
\usepackage{enumitem}
\usepackage[round]{natbib}

\usepackage{geometry}

\newtheorem{theorem}{Theorem}
\newtheorem{definition}{Definition}
\newtheorem{remark}{Remark}

\newtheorem{lemma}{Lemma}

\newenvironment{proof}{\noindent{\bf Proof.}}
{\hspace*{\fill}$\Box$\par\vspace{4mm}}

\newcommand{\PAPERINO} {\texttt{RecognizePG}\xspace}
\newcommand{\PAPERINObis} {\texttt{RecognizeDPG}\xspace}
\newcommand{\U} {\text{Upper}}
\newcommand{\Cross} {\text{Cross}}
\newcommand{\NULL} {\texttt{NULL}}

\newcommand{\Colored} {\text{Colored}}
\newcommand{\Blank} {\text{Blank}}

\author[Lorenzo Balzotti]{Lorenzo Balzotti\affiliationmark{1}}
\title[Simpler and Unified Recognition Algorithm for PGs and DPGs]{Simpler and Unified Recognition Algorithm for Path Graphs and Directed~Path~Graphs}
\affiliation{
  Department of Statistical Sciences, Sapienza University of Rome, Rome, Italy.}

\begin{document}
\publicationdata{vol. 27:2}{2025}{9}{10.46298/dmtcs.13355}{2024-04-06; 2024-04-06; 2024-11-15}{2025-03-12}

\maketitle

\begin{abstract}
\vspace{2.5mm}
A path graph is the intersection graph of paths in a tree. A directed path graph is the intersection graph of paths in a directed tree. Even if path graphs and directed path graphs are characterized very similarly, their recognition algorithms differ widely. We further unify these two graph classes by presenting the first recognition algorithm for both path graphs and directed path graphs. We deeply use a recent characterization of path graphs, and we extend it to directed path graphs. Our algorithm does not require complex data structures and has an easy and intuitive implementation, simplifying recognition algorithms for both graph classes.
\end{abstract}

\section{Introduction}

A \emph{path graph} is the intersection graph of paths in a tree. A \emph{directed path graph} is the intersection graph of paths in a directed tree. In this article we present a recognition algorithm for  both path graphs and directed path graphs.


Path graphs are introduced by~\cite{renz}, who also gives a combinatorial, non-algorithmic characterization. The second characterization is due to~\cite{gavril_UV_algorithm}, and it leads him to present a first recognition algorithm with $O(n^4)$ time complexity (in this paper, the input graph has $n$ vertices and $m$ edges). The third characterization is due to~\cite{mew}, and it is used by Sch\"{a}ffer to build a faster recognition algorithm~\cite{schaffer}, that has $O(p(m+n))$ time complexity (where $p$ is the number of \emph{cliques}, namely, maximal induced complete subgraphs). Later, \cite{chaplick} gives a recognition algorithm with the same time complexity that uses PQR-trees. L\'{e}v\^{e}que, Maffray and Preissmann present the first characterization by forbidden subgraphs~\cite{bfm}, while, recently, Apollonio and Balzotti give another characterization~\cite{ab}, that builds on~\cite{mew}. Another algorithm is proposed in~\cite{dah} and claimed to run in $O(m+n)$ time, but it has only appeared as an extended abstract and is not considered to be complete or correct (see comments in~[\cite{chaplick}, Section 2.1.4]).

Directed path graphs are characterized by~\cite{gavril_DV_algorithm}, in the same article he also gives the first recognition algorithm that has $O(n^4)$ time complexity. In the article cited above, \cite{mew} give the second characterization of directed path graphs, that yields a recognition algorithm with $O(n^2m)$ time complexity. \cite{chaplick-gutierrez} present a linear time algorithm able to establish if a path graph is a directed path graph (actually, their algorithm requires the \emph{clique path tree} of the input graph, we refer to Section~\ref{section:old_characterizations} for further details). This implies that the algorithms in~\cite{chaplick,schaffer} can be used to obtain a recognition algorithm for directed path graphs with the same time complexity. At the state of art, this technique leads to the fastest algorithms.

Path graphs and directed path graphs are classes of graphs between \emph{interval graphs} and \emph{chordal graphs}. A \emph{hole} is is a chordless cycle of length at least four.  A graph is a chordal graph if it does not contain a \emph{hole} as an induced subgraph. \cite{gavril1} proves that a graph is chordal if and only if it is the intersection graph of subtrees of a tree. We can recognize chordal graphs in $O(m+n)$ time~\cite{sc18,sc19}.

A graph is an interval graph if it is the intersection graph of a family of intervals on the real line; or, equivalently, the intersection graph of a family of subpaths of a path. \cite{lekkerkerker-boland} characterize interval graphs as chordal graphs with no asteroidal triples, where an asteroidal triple is a stable set of three vertices such that each pair is connected by a path avoiding the neighborhood of the third vertex. Interval graphs can be recognized in linear time by several algorithms~\cite{booth-lueker,corneil-olariu,habib-mcconnell,hsu,hsu-mcconnell,korte-mohring,mcconnel-spinrad}.

We now introduce a last class of intersection graphs. A \emph{rooted path graph} is the intersection graph of directed paths in a rooted tree. Rooted path graphs can be recognized in linear time by using the algorithm by~\cite{dietz}. All inclusions among the introduced classes of graphs are summarized in the following chain of strict inclusion relations:
\begin{equation*}
\text{interval graphs $\subset$ rooted path graphs $\subset$ directed path graphs $\subset$ path graphs $\subset$ chordal graphs}.
\end{equation*}

\paragraph{Our contribution} 
In this article we present the first recognition algorithm for both path graphs and directed path graphs, it has $O(p(m+n))$ time complexity.
Our algorithm is based on the characterization of path graphs and directed path graphs given by~\cite{mew}, and we strictly use the recent characterization by~\cite{ab} that simplifies the one in~\cite{mew}.

On the side of path graphs, we believe that compared to~\cite{chaplick,schaffer}, our algorithm provides a simpler and very shorter treatment (the whole explanation is in Section~\ref{section:recognition_algorithm_e_correctness}). Moreover, it does not need complex data structures while the algorithm in~\cite{chaplick} is based on PQR-trees and the algorithm in~\cite{schaffer} is a complex backtracking algorithm.

On the side of directed path graphs, at the state of art, our algorithm is the only one that does not use the results in~\cite{chaplick-gutierrez}, in which it is given a linear time algorithm able to establish whether a path graph is a directed path graph too (see Theorem~\ref{th:il_piu_importante} for further details). Thus, prior to this paper, it was necessary to implement two algorithms to recognize directed path graphs: a recognition algorithm for path graphs as in~\cite{chaplick,schaffer}, and the algorithm in~\cite{chaplick-gutierrez} that in linear time is able to determining whether a path graph is also a directed path graph. Instead, we obtain our recognition algorithm for directed path graphs by slightly modifying the recognition algorithm for path graphs. In this way, we do not improve the running time, rather we provide a simpler algorithm using known characterizations.


\paragraph{Our approach}

The recognition algorithm \PAPERINO for path graph  is mainly built on path graphs' characterization in~\cite{ab}. This characterization decomposes the input graph $G$ by clique separators as in~\cite{mew}, then at the recursive step one has to find a proper vertex coloring of an antipodality graph satisfying some particular conditions; see Section~\ref{section:ab} for the definition of \emph{clique separator} and \emph{antipodality graph}, and Theorem~\ref{th:coloring} for the characterization. In a few words, an antipodality graph has as vertex set some subgraph of $G$, and two vertices are connected if the corresponding subgraphs of $G$ are antipodal. Building all the antipodality graphs by brute force requires more time than the overall complexity of algorithms in~\cite{chaplick,schaffer}. We overcome this problem by visiting the connected components in a smart order. This order allows us to establish all the antipodality relations faster. This is done in Step~\ref{al:inizialization}, Step~\ref{al:crossing}, and Step~\ref{al:antipodal} that are the core of algorithm \PAPERINO.

On the side of directed path graphs, we first extend the characterization in~\cite{ab} for path graphs to directed path graphs, and then we adapt the recognition algorithm for path graphs to directed path graphs, obtaining algorithm \PAPERINObis.

We stress that the algorithm \PAPERINO and algorithm \PAPERINObis have minimal differences, and they could be both merged in the same algorithm. However, for enhanced readability and clarity, we prefer to separate them.

\paragraph{Organization} The paper is organized as follows. In Section~\ref{section:old_characterizations} we present the characterization of path graphs and directed path graphs given by~\cite{mew}, while in Section~\ref{section:ab} we explain the characterization of path graphs by~\cite{ab}. In Section~\ref{section:recognition_algorithm_e_correctness} we present our recognition algorithm for path graphs, we prove its correctness, we present some implementation details and we compute its time complexity. Finally, in Section~\ref{section:algorithm_DPG} we provide a similar analysis for directed path graphs.

\section{Earlier characterizations of path graphs and directed path graphs}\label{section:old_characterizations}

In this section, we present the characterization of path graphs and directed path graphs as described in~\cite{mew}. We start with a formal definition of these classes of graphs.

We denote by $G=(V,E)$ a finite connected undirected graph, where $V$, $|V|=n$, is a set of \emph{vertices} and $E$, $|E|=m$, is a collection of pairs of vertices called \emph{edges}. Let ${P}$ be a finite family of nonempty sets. The intersection graph of ${P}$ is obtained by associating each set in ${P}$ with a vertex and by connecting two vertices with an edge exactly when their corresponding sets have a nonempty intersection. The intersection graph of a family of paths in a tree is called \emph{path graph}. The intersection graph of a family of directed paths in a directed tree is called \emph{directed path graph}. We say that two directed or undirected paths intersect if and only if they have at least one vertex in common.

The first characterizations of path graphs and directed path graphs are due to~\cite{gavril_DV_algorithm,gavril_UV_algorithm}.
We let $\mathbf{C}$ denote the set of cliques of $G$, and for every $v\in V(G)$ let $\mathbf{C}_v=\{C\in\mathbf{C}\ |\ v\in C\}$. We recall that a clique is a maximal induced complete subgraph. Moreover, for a graph $G$ and for a subset $A$ of $V(G)$, we denote the graph induced by $A$ in $G$ by $G[A]$.

\begin{theorem}[\cite{gavril_DV_algorithm,gavril_UV_algorithm}]\label{th:Gavril_UV_DV}
A graph $G=(V,E)$ is a path graph (resp.~directed path graph) if and only if there exists a tree $T$ (resp.~directed tree $T$) with vertex set $\mathbf{C}$, such that for every $v\in V$, $T[\mathbf{C}_v]$ is a path (resp.~directed path) in $T$.
\end{theorem}

The tree $T$ of the previous theorem is called the \emph{clique path tree of $G$} if $G$ is a path graph or the \emph{directed clique path tree of $G$} if $G$ is a directed path graph. In Figure~\ref{fig:examples_Path_Graph_and_cpt}, the left part shows a path graph $G$, and on the right there is a clique path tree of $G$. Symmetrically, in Figure~\ref{fig:examples_Directed_Path_Graph_and_dcpt}, the left part shows a directed path graph $G$, and on the right there is a directed clique path tree of $G$. 

\begin{figure}[h]
\centering
	\begin{subfigure}{5cm}
\begin{overpic}[width=5cm]{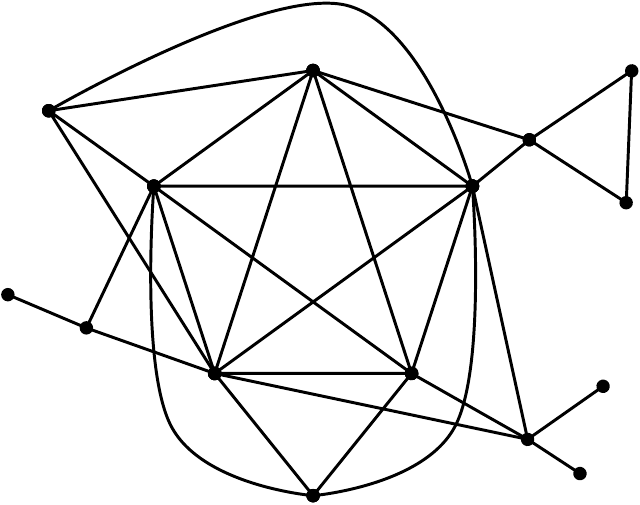}
\put(23,75){$G$}
\put(47,70){$1$}
\put(76.7,45.5){$2$}
\put(67,20){$3$}
\put(30.5,13){$4$}
\put(23,52.5){$5$}
\put(81,60){$6$}
\put(78.5,3){$7$}
\put(46.9,-6){$8$}
\put(11,20){$9$}
\put(0,64){$10$}
\put(95,70){$11$}
\put(94.5,39){$12$}
\put(92,21){$13$}
\put(88,-3){$14$}
\put(-4,35){$15$}
\end{overpic}
\end{subfigure}
\qquad\qquad\qquad
	\begin{subfigure}{5cm}
\begin{overpic}[width=4cm,percent]{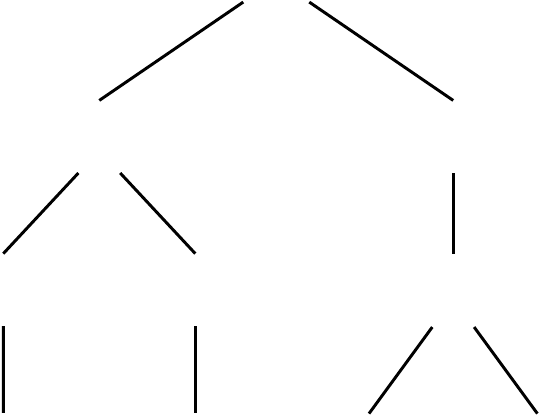}
\put(27,80){$\{1,2,3,4,5\}$}
\put(-7,48){$\{1,2,4,5,10\}$}
\put(59,48){$\{2,3,4,5,8\}$}
\put(-15,20){$\{1,2,6\}$}
\put(21,20){$\{4,5,9\}$}
\put(64,20){$\{2,3,4,7\}$}
\put(-20,-9){$\{6,11,12\}$}
\put(23,-9){$\{9,15\}$}
\put(55,-9){$\{7,13\}$}
\put(88,-9){$\{7,14\}$}
\end{overpic}
\end{subfigure}
\vspace{3mm}
\caption{a path graph $G$ (on the left) and a clique path tree of $G$ (on the right).}
\label{fig:examples_Path_Graph_and_cpt}
\end{figure}

\begin{figure}[h]
\centering
	\begin{subfigure}{3.5cm}
\begin{overpic}[width=3cm,percent]{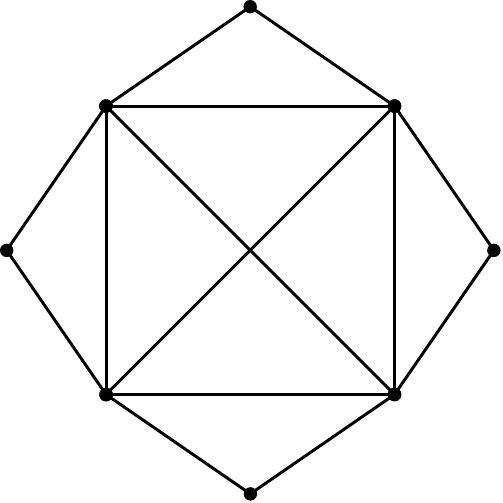}
\put(29.5,92.5){$G$}

\put(14,81){1}
\put(80,81){2}
\put(81,14){3}
\put(13,13){4}
\put(47.5,-10){5}
\put(100.5,50){6}
\put(48,102){7}
\put(-7.5,50){8}

\end{overpic}
\end{subfigure}
\qquad\qquad\qquad
	\begin{subfigure}{3.6cm}
\begin{overpic}[width=3.6cm,percent]{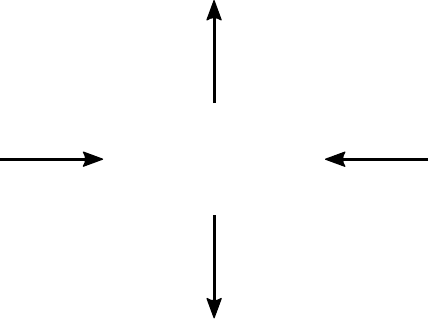}
\put(28,34.52){$\{1,2,3,4\}$}

\put(32,-9){$\{3,4,5\}$}
\put(32,78){$\{1,2,7\}$}

\put(-37,34.52){$\{1,4,8\}$}
\put(101,34.52){$\{2,3,6\}$}
\end{overpic}
\end{subfigure}
\vspace{1mm}
\caption{a directed path graph $G$ (on the left) and a directed clique path tree of $G$ (on the right).}
\label{fig:examples_Directed_Path_Graph_and_dcpt}
\end{figure}

Theorem~\ref{th:Gavril_UV_DV} specializes the celebrated characterization of chordal graphs, still due to~\cite{gavril1}, as those graphs possessing a \emph{clique tree} as stated below.

\begin{theorem}[\cite{gavril1}]\label{thm:Gavril_chordal}
A graph $G$ is a chordal graph if and only if it there exists a tree $T$, called a \emph{clique tree}, with vertex set $\mathbf{C}$ such that, for every $v\in V$, $T[\mathbf{C}_v]$ is a tree in $T$.
\end{theorem}

The main goal of our paper is: given a graph $G$, to find a (directed) clique path tree of $G$ or to say that $G$ is not a (directed) path graph. To achieve our goal, we adopt the same approach as in~\cite{mew}, by decomposing $G$ recursively by clique separators.

\cite{mew} characterized several classes of intersection families, two of which are path graphs and directed path graphs. Now we present some of their theorems and lemmata that concern path graphs and directed path graphs.

A clique is a \emph{clique separator} if its removal disconnects the graph in at least two connected components. A graph with no clique separator is called \emph{atom}. For example, every cycle has no clique separator, and the butterfly/hourglass graph (the graph obtained by joining two copies of the cycle graph $C_3$ with a common vertex) has two cliques and it is an atom. In~\cite{mew} it is proved that an atom is a path graph and/or a directed path graph if and only if it is a chordal graph; moreover, every chordal graph that is an atom has at most two cliques.

From now on, let us assume that a clique $C$ separates $G=(V,E)$ into subgraphs $\gamma_i=G[C\cup V_i]$, $1\leq i\leq s$, $s\geq2$. Let $\Gamma_C=\{\gamma_1,\ldots,\gamma_s\}$.

\cite{mew} defined the following binary relations on $\Gamma_C$. A clique $K$ of a member $\gamma$ in $\Gamma_C$ is called a  \emph{relevant clique} if $K \cap C\not=\emptyset$.
\begin{itemize}\itemsep0em
\item \emph{Attachedness}, denoted by $\bowtie$ and defined as $\gamma\bowtie \gamma'$ if and only if there is a relevant clique $K$ of $\gamma$ and a relevant clique $K'$ of $\gamma'$ such that $K\cap K'\neq\emptyset$.
\item \emph{Dominance}, denoted by $\leq$ and defined as $\gamma\leq \gamma'$ if and only if $\gamma\bowtie \gamma'$ and for each relevant clique $K'$ of $\gamma'$ either $K\cap C\subseteq K'\cap C$ or $K\cap K'\cap C=\emptyset$ for each relevant clique $K$ of $\gamma$.
\item \emph{Antipodality}, denoted by $\leftrightarrow$ and defined as $\gamma \leftrightarrow\gamma'$ if and only if there are relevant cliques $K'$ of $\gamma'$ and $K$ of $\gamma$ such that $K\cap K'\cap C\not=\emptyset$ and $K\cap C$ and $K'\cap C$ are inclusion-wise incomparable.
\end{itemize}

We observe that if $\gamma\bowtie\gamma'$, then either $\gamma\leq\gamma'$, or $\gamma'\leq\gamma'$ or $\gamma\leftrightarrow\gamma'$. To better understand these relations we use Figure~\ref{fig:pentagono_colorato}, that shows a clique separator $C$ of $G$ of Figure~\ref{fig:examples_Path_Graph_and_cpt} and its connected components. We stress that $C=\{1,2,3,4,5\}$, $\Gamma_C=\{\gamma_1,\gamma_2,\gamma_3,\gamma_4,\gamma_5\}$, where $\gamma_1=\{1,2,6,11,12\}$, $\gamma_2=\{2,3,4,7,13,14\}$, $\gamma_3=\{2,3,4,5,8\}$, $\gamma_4=\{4,5,9,15\}$ and $\gamma_5=\{1,2,4,5,10\}$. Clearly, $\gamma_1,\ldots,\gamma_5$ are path graphs. It holds that $\gamma_5\geq\gamma_1$, $\gamma_5\geq\gamma_4$, $\gamma_3\geq\gamma_2$, $\gamma_3\geq\gamma_4$, $\gamma_1\leftrightarrow\gamma_2$, $\gamma_1\leftrightarrow\gamma_3$, $\gamma_3\leftrightarrow\gamma_5$, $\gamma_2\leftrightarrow\gamma_5$, $\gamma_2\leftrightarrow\gamma_4$, and the relation  $\bowtie$ follows from above.

\begin{figure}[h]
\centering
\begin{subfigure}{5cm}
\begin{overpic}[width=5cm,percent]{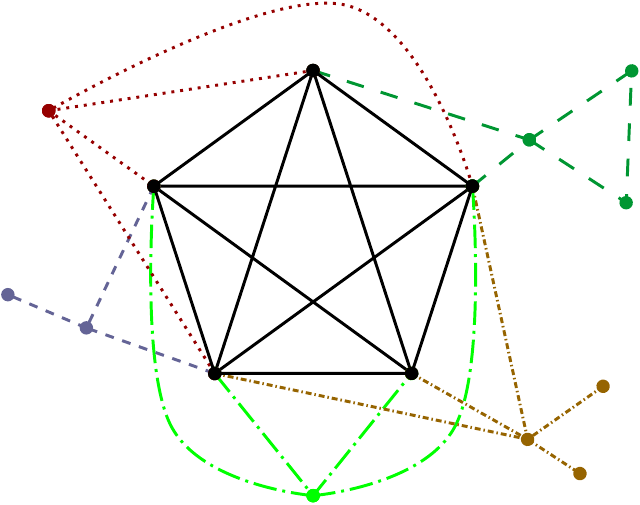}

\put(100,55){$\gamma_1$}
\put(93,10){$\gamma_2$}
\put(58,-2){$\gamma_3$}
\put(0,23.5){$\gamma_4$}
\put(24,75){$\gamma_5$}
\put(46,40){$C$}

\put(15,70){$G$}
\put(47,70){$1$}
\put(76.7,45.5){$2$}
\put(67,20){$3$}
\put(30.5,13.5){$4$}
\put(23,53){$5$}
\put(81,60){$6$}
\put(78.5,3){$7$}
\put(46.9,-6){$8$}
\put(11,20){$9$}
\put(0,64){$10$}
\put(95,70){$11$}
\put(94.5,39){$12$}
\put(92,21){$13$}
\put(88,-3){$14$}
\put(-4,35){$15$}

\end{overpic}
\end{subfigure}
\caption{a clique separator $C$ in $G$ of Figure~\ref{fig:examples_Path_Graph_and_cpt}. The connected components $\gamma_1,\ldots,\gamma_5$ are highlighted by different colors and hatchings.}
\label{fig:pentagono_colorato}
\end{figure}
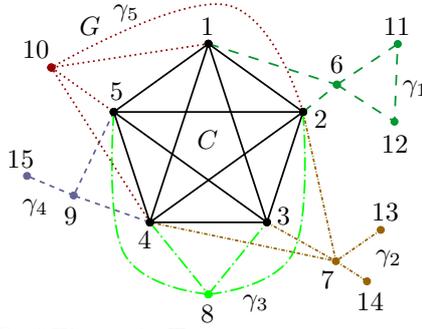

Now we present the two theorems in~\cite{mew} that characterize path graphs and directed path graphs. A subgraph $\gamma_i$ is a neighboring subgraph of $v\in V(C)$ if some vertex in $V_i$ is adjacent to $v$.  For a function $f$ and a set $X$ we define $f(X)=\bigcup_{x\in X}f(x)$. The following theorem characterizes directed path graphs.

\begin{theorem}[\cite{mew}]\label{mew1}
A chordal graph $G$ is a path graph if and only if $G$ is an atom or for a clique separator $C$ each graph $\gamma\in \Gamma_C$ is a path graph and there exists $f:\Gamma_C\rightarrow [s]$ such that:
\renewcommand{\theenumi}{\thetheorem.\arabic{enumi}}
\begin{enumerate}
\item\label{com:mw_i} if $\gamma\leftrightarrow \gamma'$, then $f(\gamma)\not=f(\gamma')$;
\item\label{com:mw_ii} if $\gamma$, $\gamma'$ and $\gamma''$ are neighboring subgraphs of $v$, for some vertex $v\in C$, then $|f(\{\gamma,\gamma',\gamma''\})|\leq 2$.
\end{enumerate}
\end{theorem}

\begin{theorem}[\cite{mew}]\label{mew2}
A chordal graph $G$ is a directed path graph if and only if $G$ is an atom or for a clique separator $C$ each graph $\gamma\in \Gamma_C$ is a path graph and the $\gamma_i$'s can be 2-colored such that no antipodal pairs have the same color.
\end{theorem}

Note that the recursive step in Theorem~\ref{mew1} and Theorem~\ref{mew2} is a coloring problem: we have to find a proper vertex coloring $f$ of the \emph{antipodality graph} $A_C$, that is the graph on $\Gamma_C$ induced by $\leftrightarrow$, formally $A_C=(\Gamma_C,\{\gamma\gamma' \,|\, \gamma,\gamma'\in \Gamma_C \text{ and } \gamma\leftrightarrow\gamma'\})$.

\cite{chaplick-gutierrez} gave an algorithm that, by starting from the clique path tree of a path graph, builds the directed clique path tree, if it exists, in linear time. This result is summarized in the following theorem.

\begin{theorem}[\cite{chaplick-gutierrez}]\label{th:il_piu_importante}
If there exists a polynomial algorithm that tests if a graph $G$ is a path graph and returns a clique path tree of $G$ when the answer is ``yes", then there exists an algorithm with the same complexity to test if a graph is a directed
path graph.
\end{theorem}

Theorem~\ref{th:il_piu_importante} implies that the algorithms in~\cite{chaplick,schaffer} can be extended to algorithms able to recognize directed path graphs with the same time complexity.

\section{A recent characterization of path graphs}\label{section:ab}

In this section we introduce some results and notations in~\cite{ab}, that give a new characterization of path graphs presented in Theorem~\ref{th:coloring}. Indirectly, some of these results allow us to efficiently recognize directed path graphs too (see Section~\ref{section:algorithm_DPG} and Theorem~\ref{th:coloring_DPG}).

For the following lemmata and discussions, let us fix a clique separator $C$ of $G$. We use round brackets to denote ordered sets. For $\ell\in \mathbb{N}$, we denote by $[\ell]$ the interval $\{1,\ldots,\ell\}$.

Let $\sim$ be the equivalence relation defined on $\Gamma_C$ by $\gamma\sim\gamma'\Leftrightarrow(\gamma\leq\gamma'\wedge\gamma'\leq\gamma)$, for all $\gamma,\gamma'\in\Gamma_C$. Moreover, given $\gamma\in\Gamma_C$, we denote by $[\gamma]_\sim$ the equivalence class of $\gamma$ w.r.t. $\sim$, i.e., $[\gamma]_\sim=\{\gamma'\in\Gamma_C \,|\, \gamma'\sim\gamma\}$. We say that $[\gamma]_\sim\leftrightarrow[\gamma']_\sim$ if and only if $\eta\leftrightarrow\eta'$ for any $\eta\in[\gamma]_\sim$ and $\eta'\in[\gamma']_\sim$; moreover, we say that $[\gamma]_\sim$ is a neighboring of $v$ if and only if $\eta$ is a neighboring of $v$, for any $\eta\in[\gamma]_\sim$. In the following lemma, see Lemma 2 in~\cite{ab}, we state that it is not restrictive to work in $\Gamma_C/\!\!\sim$.

\begin{lemma}[\cite{ab}]\label{lemma:tilde}
If there exists $f:\Gamma_C/\!\!\sim\to[\ell]$ satisfying conditions~\ref{com:mw_i} and~\ref{com:mw_ii}, then $\widetilde{f}:\Gamma_C\to[\ell]$ defined by $\widetilde{f}(\gamma)=f([\gamma]_\sim)$ satisfies conditions~\ref{com:mw_i} and~\ref{com:mw_ii}.
\end{lemma}

From now on, unless otherwise noted, we assume that $\Gamma_C=\Gamma_C/\!\!\sim$. Even if this works from a theoretical point of view, we need to compute the classes of equivalence w.r.t. $\sim$ in our algorithms. 

We define $\U_C=\{u\in\Gamma_C \,|\, u\not\leq\gamma$, for all $u\neq\gamma\in\Gamma_C\}$. We note that $\U_C$ is the set of ``upper bounds" of $\Gamma_C$ w.r.t. $\leq$. From now on we fix an ordering $(u_1,u_2,\ldots,u_r)$  of $\U_C$. For all $i<j\in[r]$ we define

\begin{equation}\label{dgamma}	
D_{i}^C=\{\gamma\in \Gamma_C \ |\ \gamma\leq u_i  \text{ and } \gamma\nleq u_j,\,\,\forall j\in[r]\setminus \{i\}\},
\end{equation}
\begin{equation}\label{dgammagamma'}
D_{i,j}^C=\{\gamma\in\Gamma_C \ |\ \gamma\leq u_i,\gamma\leq u_j  \text{ and } \gamma\nleq u_k,\,\,\forall k\in[r]\setminus \{i,j\}\},
\end{equation}	
\begin{equation}\label{dmathcal}
\mathcal{D}^C=\Big(\bigcup_{i\in[r]}D^C_i\Big) \cup \Big(\bigcup_{i<j\in[r]}D^C_{i,j}\Big).
\end{equation}

If it is clear from the context, then we omit the superscript $C$. Note that $D_i\neq\emptyset$ for all $i\in[r]$, indeed $u_i\in D_i$. However, $D_{i,j}$ can be empty for some $i<j\in[r]$.

Given $\gamma,\gamma',\gamma''\in\Gamma_C$, we say that $\{\gamma,\gamma',\gamma''\}$ is an \emph{antipodal triangle} if $\gamma,\gamma',\gamma''$ are pairwise antipodal; moreover, if $\gamma,\gamma',\gamma''$ are also neighbouring of $v$, for some $v\in C$, then we say that $\{\gamma,\gamma',\gamma''\}$ is a \emph{full antipodal triangle}. We note that if $G$ is a path graph, then Theorem~\ref{mew1} implies that $\Gamma_C$ has no full antipodal triangle, for all clique separator $C$  of $G$.

In order to better understand the characterization of path graphs in~\cite{ab}, we provide the following lemma (see Remark 1 and Lemma 4 in~\cite{ab}), explaining that the absence of full antipodal triangle in $\U_C$ implies a rigid structure to the antipodality graph $A_C$. In particular, for any $D\in D^C$, we know how antipodality works between elements in $D$ and elements not in $D$.

\begin{lemma}[\cite{ab}]\label{lemma:elenco_trivial_2}
Let $C$ be a clique separator of $G$ and let $i<j\in[r]$. If $\U_C$ has no full antipodal triangle, then the following statements hold:
\renewcommand{\theenumi}{\thetheorem.\arabic{enumi}}
\begin{enumerate}[label=(\alph*), ref=\thelemma.(\alph*)]\itemsep0em
\item\label{item:partition} $\mathcal{D}$ forms a partition of $\Gamma_C$,
\item\label{item:antipodal_k_k'_1} $\gamma\leftrightarrow\gamma'$, $\gamma\in D_{i,j}$ and $\gamma'\not\in D_{i,j}$ $\Rightarrow$ $\gamma'\in D_i\cup D_j$,
\item\label{item:antipodal_D_i_Upper} $\gamma\leftrightarrow\gamma'$, $\gamma\in D_i$ and $\gamma'\not\in D_i$ $\Rightarrow\gamma\leftrightarrow u_k$ for $\gamma'\leq u_k$ and $k\neq i$.
\end{enumerate}
\end{lemma}

Referring to Figure~\ref{fig:pentagono_colorato}, up to permutations, $\U_C=(u_1,u_2)=(\gamma_3,\gamma_5)$, $\mathcal{D}^C=\{D_1,D_2,D_{1,2}\}$, $D_{1}=\{\gamma_2,\gamma_3\}$, $D_{2}=\{\gamma_1,\gamma_5\}$ and $D_{1,2}=\{\gamma_4\}$.

The characterization of path graphs given in~\cite{ab} is summarized in the following theorem (see Definition 2, Definition 3, and Theorem 4 in~\cite{ab}). Note that, as for Theorem~\ref{mew1}, in the recursive step we have to find a proper coloring $f$ of $A_C$.

\begin{theorem}[\cite{ab}]\label{th:coloring}
A chordal graph $G$ is a path graph if and only if $G$ is an atom or for a clique separator $C$ each graph $\gamma\in \Gamma_C$ is a path graph, $\U_C=(u_1,u_2,\ldots,u_r)$ has no full antipodal triangle and there exists $f:\Gamma_C\rightarrow [r+1]$ such that:
\renewcommand{\theenumi}{\thetheorem.\arabic{enumi}}
\begin{enumerate}
\item\label{item:col_2} for all $i\in[r]$, $f(u_i)=i$, 
\item\label{item:col_3} for all $i\in[r]$, $f(D_i)\subseteq\{i,r+1\}$,
\item\label{item:col_4} for all $i<j\in[r]$, $f(D_{i,j})\subseteq\{i,j\}$,
\item\label{item:col_5} for all $i\in[r]$, for all $\gamma\in D_i$, if $\exists u\in \U_C$ such that $\gamma\leftrightarrow u$, then $f(\gamma)=i$,
\item\label{item:col_6} for all $i<j\in[r]$, for all $\gamma\in D_{i,j}$ such that $\exists\gamma'\in D_{k}$, for $k\in\{i,j\}$, satisfying $\gamma\leftrightarrow\gamma'$, then $f(\gamma)=\{i,j\}\setminus\{k\}$,
\item\label{item:col_7} for all $D\in\mathcal{D}^C$, for all $\gamma,\gamma'\in D$ such that $\gamma\leftrightarrow\gamma'$, $f(\gamma)\neq f(\gamma')$.
\end{enumerate}
\end{theorem}

Let us comment briefly on Theorem~\ref{th:coloring}. We have to check for the absence of a full antipodal triangle in $\U_C$, this is an easy request because we are working on upper bounds. The first  condition sets the colors of elements in $\U_C$ and the second and third conditions state which colors we can use to color elements in $D$, for every $D\in\mathcal{D}^C$; note that every $D\in\mathcal{D}^C$ is 2-colored by $f$. To better understand the fourth and fifth conditions let $\Cross_C$ be the set of elements that are antipodal to elements belonging to different sets of partition $\mathcal{D}^C$, formally $\Cross_C=\{\gamma\in\Gamma_C \,|\, \gamma\in D$ for some $D\in\mathcal{D}^C$ and there exists $\gamma'\not\in D$ satisfying $\gamma\leftrightarrow\gamma'\}$. The fourth and fifth conditions explain how to color elements in $\Cross_C$. The last condition says that two antipodal elements in $D$, for any $D\in\mathcal{D}^C$, have different colors under $f$. Finally, we observe that in~\cite{ab} $u$ in condition~\ref{item:col_5} satisfies $u\not\in D_i$, but, for algorithmic convenience, we choose $u\in\U_C$ because of Lemma~\ref{lemma:elenco_trivial_2}.

\section{Recognition algorithm for path graphs}\label{section:recognition_algorithm_e_correctness}

In this section we introduce algorithm \PAPERINO, that is able to recognize path graphs. In Subsection~\ref{sub:correctness} we present the algorithm and we prove its correctness. In Subsection~\ref{section:implementation_details_&_running} we compute its time complexity.

\subsection{The algorithm and its correctness}\label{sub:correctness}

We present the algorithm \PAPERINO. Note that it is an implementation of Theorem~\ref{th:coloring} with very small changes. W.l.o.g., we assume that $G$ is connected, indeed a graph $G$ is a path graph if and only if all its connected components are path graphs. Moreover, we can obtain the clique path tree of $G$ by merging arbitrarily the clique path tree of every connected component.

\mybreak
\textbf{\PAPERINO}\\
\textbf{Input:} a graph $G$\\
\textbf{Output:} if $G$ is a path graph, then return a clique path tree of $G$; else QUIT
\begin{enumerate}
\item\label{al:chord} Test if $G$ is chordal. If not, then QUIT.

\item\label{al:Q} If $G$ has at most two cliques, then return a clique path tree of $G$. Else find a  clique separator $C$, let $\Gamma_C=\{\gamma_1,\ldots,\gamma_s\}$, $\gamma_i=G[V_i\cup C]$, be the set of connected components separated by $C$.

\item\label{al:rec} Recursively test the graphs $\gamma\in \Gamma_C$. If any one is not a path graph, then QUIT, otherwise, return a clique path tree $T_\gamma$ for each $\gamma\in\Gamma_C$. 

\item\label{al:inizialization} Compute $\Gamma_C/\!\!\sim$ and initialize $f(\gamma)=$ \NULL, for all $\gamma\in\Gamma_C/\!\!\sim$. Compute $\U_C$ and fix an order of its element, i.e., $\U_C=(u_1,\ldots,u_r)$, and set $f(u_i)=i$, for all $i\in[r]$. If a full antipodal triangle in $\U_C$ is found, then QUIT. Compute $D_i$, for all $i\in[r]$, and $D_{i,j}$, for $i<j\in[r]$.

\item\label{al:crossing}
\begin{itemize}[leftmargin=1em]
\item For all $i\in[r]$, if there exist $\gamma\in D_i$ and $u\in\U_C$ such that $\gamma\leftrightarrow u$, then $f(\gamma)=i$.

 \item For all $i<j\in[r]$,
	\begin{itemize}\itemsep0em
	\item if there exist $\gamma\in D_{i,j}$, $\gamma'\in D_i$, $\gamma''\in D_j$, such that $\gamma\leftrightarrow\gamma'$ and $\gamma\leftrightarrow\gamma''$, then QUIT,
	\item if there exist $\gamma\in D_{i,j}$ and $\gamma'\in D_j$ such that $\gamma\leftrightarrow\gamma'$, then $f(\gamma)=i$,
	\item if there exist $\gamma\in D_{i,j}$ and $\gamma'\in D_i$ such that $\gamma\leftrightarrow\gamma'$, then $f(\gamma)=j$.
	\end{itemize}

\end{itemize}

\item\label{al:antipodal} 
\begin{itemize}[leftmargin=1em]

\item For all $i\in[r]$, extend $f$ to all elements in $D_i$ so that $f(D_i)\subseteq\{i,r+1\}$ and $f(\gamma)\neq f(\gamma')$ for all $\gamma,\gamma'\in D_i$ satisfying $\gamma\not\leftrightarrow\gamma'$. If it is not possible, then QUIT.

\item  For all $i<j\in[r]$, extend $f$ to all elements in $D_{i,j}$ so that $f(D_{i,j})\subseteq\{i,j\}$ and $f(\gamma)\neq f(\gamma')$ for all $\gamma,\gamma'\in D_{i,j}$ satisfying $\gamma\not\leftrightarrow\gamma'$. If it is not possible, then QUIT.
\end{itemize}

\item\label{al:ct} Convert the coloring $f:\Gamma_C/\!\!\sim\rightarrow[r+1]$ in a clique path tree of $\Gamma_C$.
\end{enumerate}

\begin{theorem}
\label{th:correctness}
Given a graph $G$, algorithm \PAPERINO can establish whether $G$ is a path graph. If so, algorithm \PAPERINO returns a clique path tree of $G$.
\end{theorem}
\begin{proof}
The first three steps of algorithm \PAPERINO are implied by the first part of Theorem~\ref{th:coloring}. By following Theorem~\ref{th:coloring}, we have to check that there are no full antipodal triangle in $\U_C$ (this is performed in Step~\ref{al:inizialization}), and we have to find $f:\Gamma_C\rightarrow[r+1]$ satisfying~\ref{item:col_2},\ldots,\ref{item:col_7}, where $r=|\U_C|$. This latter part is done in Step~\ref{al:inizialization}, Step~\ref{al:crossing} and Step~\ref{al:antipodal}. In particular~\ref{item:col_2} is done in Step~\ref{al:inizialization}, \ref{item:col_5} and~\ref{item:col_6} are achieved in Step~\ref{al:crossing}, and~\ref{item:col_3}, \ref{item:col_4} and~\ref{item:col_7} are reached in Step~\ref{al:antipodal}. Note that the first condition in the second case of Step~\ref{al:crossing} is indirectly present in Theorem~\ref{th:coloring}: if it happens, then we cannot satisfy condition~\ref{item:col_6} (moreover, $\gamma,\gamma',\gamma''$ would form a full antipodal triangle). Finally, Step~\ref{al:ct} completes the recursion started in Step~\ref{al:rec} by building the clique path tree on $\Gamma_C$.
\end{proof}

\subsection{Implementation details and time complexity}
\label{section:implementation_details_&_running}

In this section we analyze all steps of algorithm \PAPERINO. We want to explain them in details and precise the computational complexity of the algorithm. Some of these steps are already discussed in~\cite{schaffer}, anyway, we describe them in order to have a complete treatment. 

\subsubsection{Step~\ref{al:chord} and Step~\ref{al:Q}}
We can recognize chordal graphs in $O(m+n)$ time by using either an algorithm due to~\cite{sc18}, or an algorithm due to~\cite{sc19}. Both recognition algorithms can be extended to an algorithm that also produces a clique tree in $O(m+n)$ time~\cite{sc14}. In particular, we can list all cliques in $G$. It holds that a clique in $\mathbf{C}$ is a clique separator if and only if it is not a leaf of the clique tree.

\subsubsection{Step~\ref{al:rec}}

This step can be done by calling recursively algorithm \PAPERINO for all $\gamma\in\Gamma_C$. Obviously, Step~\ref{al:chord} has to be done only for $G$, indeed, the property to be chordal is inherited by subgraphs.

From now on, we are interested exclusively in the recursive part of algorithm \PAPERINO, i.e., from Step~\ref{al:inizialization} to Step~\ref{al:ct}. Thus we assume that $G$ is a chordal graph that is not an atom and $G$ is separated by a clique separator $C$. Let $\Gamma_C=\{\gamma_1,\ldots,\gamma_s\}$ be the set of connected components and we assume that all elements in $\Gamma_C$ are path graphs. It holds that $s\geq2$ because $C$ is a clique separator.

\subsubsection{Step~\ref{al:inizialization}}\label{sub:5,6,7}

We have to compute $\Gamma_C/\!\!\sim$, this problem is already solved in~\cite{schaffer}. Thus we first provide some definitions and results in~\cite{schaffer}. 

For any $\gamma\in\Gamma_C$, let $T_\gamma$ be the clique path tree of $\gamma$. Let $n_\gamma$ be the unique neighbour of $C$ in $T_\gamma$ (its uniqueness is proved in~\cite{mew}, in particular, it is proved that $C$ is a leaf of $T_\gamma$). Moreover, let $W(\gamma)=V(n_\gamma)\cap V(C)$.
\begin{definition}[\cite{schaffer}]\label{def:f}
Let $\gamma\in\Gamma_C$ and $v\in V(C)$. We define $F(\gamma,v)$ as the node of $T_\gamma$ representing the clique containing $v$ that is furthest from $C$. We observe that if $v\not\in W(\gamma)$, then $F(\gamma,v)=\emptyset$.
\end{definition}

By using Definition~\ref{def:f}, one can obtain the following lemma.

\begin{lemma}[\cite{schaffer}]
\label{lemma:geq}
Let $\gamma,\gamma'\in\Gamma_C$. It holds that
\begin{equation*}
\gamma'\leq\gamma \Leftrightarrow \gamma\cap\gamma'\neq\emptyset \text{ and } F(\gamma,v)=F(\gamma,v') \text{ for all } v,v'\in W(\gamma').
\end{equation*}
Moreover, computing if $\gamma'\leq\gamma$ and/or $\gamma'\leftrightarrow\gamma$ costs $\min(|W(\gamma)|,|W(\gamma')|)$.
\end{lemma}
\begin{remark}[\cite{schaffer}]\label{remark:f}
To compute $F(\gamma,v)$ we do one breadth-first traversal of $T_\gamma$ starting at $n_\gamma$. Each time we visit a new node $C'$, for each $v\in V(C')$, if $v\in V(C)$, we update $F(\gamma, v)$. This costs constant time for every pair $(C',v)$ such that $C'$ is a clique in $\gamma$ and $v\in V(C')$. There are at most $m+n$ such pairs.
\end{remark}

Now we can explain how to compute $\Gamma_C/\!\!\sim$. First, we sort all $\gamma\in\Gamma_C$ so that the $\gamma_i$ precedes $\gamma_j$ if $|W(\gamma_i)|\geq|W(\gamma_j)|$: we compute $W(\gamma)$ for all $\gamma\in\Gamma_C$ (it costs $|W(\gamma)|$ for every $\gamma$), then the sorting can be executed in $O(s)$ time by using bucket sort (we remember that $s\leq n$).

Now, let $\gamma,\gamma',\gamma''\in\Gamma_C$ satisfy $|W(\gamma)|=|W(\gamma')|=|W(\gamma'')|$ and $v\in W(\gamma)\cap W(\gamma') \cap W(\gamma'')$, for any $v\in C$. By Lemma~\ref{lemma:geq}, we check if $\gamma\sim\gamma'$ in $O(|W(\gamma)|)$ time. If $\gamma\not\sim\gamma'$, then either ($\gamma''\sim\gamma$ or $\gamma''\sim\gamma'$) or $\{\gamma,\gamma',\gamma''\}$ is a full antipodal triangle (this follows from definition of $\leftrightarrow$ and $W$). Hence we can compute $\Gamma_C/\!\!\sim$ in $O(\sum_{\gamma\in\Gamma_C}|W(\gamma)|)$ time, indeed every element in $\Gamma_C$ need to be checked with at most two other elements in $\Gamma_C$.

To argue with the second part of Step~\ref{al:inizialization}, let $u(v)=\{u\in \U_C \,|\, v\in W(u)\}$. Note that $|u(v)|\leq2$ for all $v\in V(C)$, otherwise $u(v)\supseteq\{u_i,u_j,u_k\}$ for some $v\in V(C)$, thus $\{u_i,u_j,u_k\}$ is a full antipodal triangle. Hence, by  Lemma~\ref{lemma:geq}, $\gamma\in \U_C$ if and only if there not exists $u_i\in \U_C$ such that $F(u_i,v)=F(u_i,v')\neq\emptyset$ for all $v,v'\in\ W(\gamma)$. Hence, to establish if $\gamma$ is in $\U_C$ it is sufficient to look at $u(v)$ for all $v\in W(\gamma)$.
By using a similar argument, we can compute $D$, for $D\in \mathcal{D}$. Hence Step~\ref{al:inizialization} can be performed in $O(\sum_{\gamma\in\Gamma_C}|W(\gamma)|)$ time.

\subsubsection{Step~\ref{al:crossing}}

\begin{lemma}
Step~\ref{al:crossing} has $O(\sum_{\gamma\in\Gamma_C}|W(\gamma)|)$ time complexity.
\end{lemma}
\begin{proof}
For the first case of Step~\ref{al:crossing}, let $\gamma\in D_i$. It suffices to check if there exists $u\in U\setminus\{u_i\}$ neighbor of $v$, for some $v\in W(\gamma)$; if so, then $u\leftrightarrow\gamma$ because of the definitions of $\U_C$ and $D_i$. Thus this check costs at most $O(|W(\gamma)|)$ time.

For the second case of Step~\ref{al:crossing}, let $\gamma\in D_{i,j}$, we have to check if there exists $\eta\in D_i$ such that $\gamma\leftrightarrow\eta$, the $D_j$'s case is symmetric. Let $v\in W(\gamma)$, note that if $\eta,\eta'\in D_i$ are neighboring of $v$ and $\eta\leftrightarrow\eta'$, then $\eta,\eta'$ and $u_j$ form a full antipodal triangle. So if $\eta\leftrightarrow\eta'$, then we fall in a QUIT in Step~\ref{al:antipodal} for $D_i$; indeed $f(\eta)=f(\eta')=i$ because of Step~\ref{al:crossing} and thus $f$ is not a 2-coloring. Hence without loss of generality with respect to the output of algorithm \PAPERINO, we can assume that $\eta\leq\eta'$ or $\eta'\leq\eta$ for all couples $\eta,\eta'\in D_i$ neighboring of $v$, for all $v\in W(u_i)\cap W(u_j)$. 

For any $v\in W(u_i)$, we define $m_i(v)$ the element in $D_i$ with minimal $W$ among all neighboring of $v$. We can compute $m_i(v)$, for all $v\in W(u_i)$, in total $O(\sum_{\eta\in D_i}|W(\eta)|)$ time by ordering elements in $D_i$ w.r.t. $W$ with bucket sort, and then visiting them in this order and updating $m_i(v)$ for all $v\in W(\eta)$ (this can be done in Step~\ref{al:crossing} in the same time complexity). 

Now let $\bar{v}\in W(\gamma)$ be so that $|W(m_i(\bar{v}))| \leq |W(m_i(v))|$ for all $v\in W(\gamma)$, we check if $\gamma\leq m_i(\bar{v})$ or $\gamma\leftrightarrow m_i(\bar{v})$. The first case implies $m_i(\bar{v})$ is a neighboring of $v'$ for all $v'\in W(\gamma)$, and thus $\gamma\leq\eta'$ for all $\eta'\in D_i$ neighboring of $v$, implying that there are not antipodal elements to $\gamma$ in $D_i$. If the second case applies, then we finished the checks for $\gamma$. In both cases we spend $O(|W(\gamma)|)$ time, and the claim follows.
\end{proof}

\subsubsection{Step~\ref{al:antipodal}}

To our goals, we introduce $f_i(\cdot)$ as $f(\cdot)$ after Step $i$ requiring that in Step $i$ algorithm \PAPERINO does not terminate in QUIT. We will use $f_5$ and $f_9$.

\begin{lemma}\label{lemma:NULL_NULL}
The following statements hold:
\begin{enumerate}
\item let $i\in[r]$ and let $\gamma,\gamma'\in D_i$ be such that $\gamma'\leq\gamma$. If $f_5(\gamma)=$ \NULL, then $f_5(\gamma')=$ \NULL ,
\item let $i<j\in[r]$ and let $\gamma,\gamma'\in D_{i,j}$ be such that $\gamma'\leq\gamma$. If $f_5(\gamma)=$ \NULL, then $f_5(\gamma')=$ \NULL.
\end{enumerate}
\end{lemma}
\begin{proof}
-- Assume by contradiction that $f_5(\gamma')\neq$ \NULL \ and $f_5(\gamma)=$ \NULL. By Step~\ref{al:crossing} there exists $i\neq k\in[r]$ such that $u_k\leftrightarrow\gamma'$, and thus $u_k,\gamma'$ are neighboring of $v$, for some $v\in C$. Thus $\gamma'$ is a neighboring of $v$ by transitivity of $\leq$. Hence it holds either $\gamma\leq u_k$, or $u_k\leq\gamma$, or $u_k\leftrightarrow\gamma$. Now, $u_k\not\leq\gamma$ because $u_k\in \U_C$, $\gamma\not\leq u_k$ because $\gamma\in D_i$ and $k\neq i$. Thus $u_k\leftrightarrow\gamma$ and Step~\ref{al:crossing} implies $f_5(\gamma)=i=f_5(\gamma')$, absurdum.

\noindent
-- Assume by contradiction that $f_5(\gamma')\neq$ \NULL  \ and $f_5(\gamma)=$ \NULL. W.l.o.g., let us assume that $f_5(\gamma')=i$. By Step~\ref{al:crossing}, there exists $\eta\in D_j$ such that $\eta\leftrightarrow\gamma'$. Thus $\eta\bowtie\gamma$ by transitivity of $\leq$. Hence, as before, it holds either $\gamma\leq\eta$, or $\eta\leq\gamma$, or $\eta\leftrightarrow\gamma$. Now, $\eta\not\leq\gamma$, otherwise $\eta\in D_{i,j}$, $\gamma\not\leq\eta$, otherwise $\gamma\leq\eta$ by transitivity of $\leq$. Thus $\eta\leftrightarrow\gamma$ and Step~\ref{al:crossing} implies $f_5(\gamma)=i=f_5(\gamma')$, absurdum.
\end{proof}


\begin{lemma}\label{lemma:cost_step_8+10}
Step~\ref{al:antipodal} has $O(\sum_{\gamma\in\Gamma_C}|W(\gamma)|)$ time complexity.
\end{lemma}
\begin{proof}
It suffices to prove that Step~\ref{al:antipodal} can be executed in $O(\sum_{\gamma\in D_{i,j}}|W(\gamma)|)$ for $D=D_{i,j}$ (second case). Indeed, the same procedure can be applied for $D=D_i$ (first case).

Let $\Colored$ be the set of elements in $D_{i,j}$ already colored before Step~\ref{al:antipodal}, i.e., $\Colored=\{\gamma\in D_{i,j} \,|\, f_5\neq\NULL\}$. We first check that there are not $\gamma,\gamma'\in\Colored$ such that $\gamma\leftrightarrow\gamma'$ and $f(\gamma)=f(\gamma')$.

We sort the elements in $D_{i,j}=(\gamma^1,\ldots,\gamma^{|D_{i,j}|})$ so that $|W(\gamma^k)|\geq |W(\gamma^{k+1})|$, for $k\in[|D_{i,j}|-1]$ (this can be done in Step~\ref{al:inizialization} without changing its time complexity). We visit the elements in $\Colored$ by following this sorting. 

For any $v\in W(D_{i,j})$ we define $\ell_i(v)$ as the lowest $\gamma\in D_{i,j}$ w.r.t $\leq$ among all visited element satisfying $f_5(\gamma)=i$. Similarly, we define $\ell_j(v)$. We initialize $\ell_i(v)=\ell_j(v)=\emptyset$ for all $v\in W(D_{i,j})$.

Let $\gamma\in\Colored$, and w.l.o.g., let us assume that $f_5(\gamma)=i$. Then $\gamma$ is not antipodal to previous visited elements colored with $i$ if $\ell_i(u)=\ell_i(v)$ for all $u,v\in W(\gamma)$. Indeed, either $\ell_i(v)=\emptyset$ for all $v\in W(\gamma)$ and thus $\gamma\not\bowtie\gamma'$ for all visited $\gamma'$ colored with $i$, or $\ell_i(v)=\gamma'$ for all $v\in W(\gamma)$ and hence there cannot exist $\gamma''$ already visited satisfying $\gamma''\leftrightarrow\gamma$ because it would imply $\ell_i(v)=\gamma''$ for some $v\in\gamma''$.

Now we deal with blank elements. We define $\Blank$ as the set of all elements in $D_{i,j}$ that do not yet have an assigned color. We say that $\gamma\in\Blank$ is \emph{solved} if there exists $u,v\in W(\gamma)$ such that $\ell_i(u)\neq\ell_i(v)$ or $\ell_j(u)\neq\ell_j(v)$. Note that if $\gamma$ is solved, then either we can set (uniquely) its color or we have to QUIT; both cases are implied by Lemma~\ref{lemma:NULL_NULL} and by the above reasoning done for $\Colored$.

For all $v\in W(D_{i,j})$ we write $\gamma\in M_v$ if $v\in W(\gamma)$, $\gamma\in\Blank$ and $|W(\gamma)|$ is maximal among all $\gamma'\in\Blank$ satisfying $v\in W(\gamma')$. Note that $|M_v|\leq2$ for all $v\in W(D_{i,j})$, otherwise there would be a full antipodal triangle.

Now we describe how to set the color of elements in $\Blank$. 
We search (in any order) $v\in W(D_{i,j})$ such that $M_v$ is solved. We observe that if we visit $M_v$ before $M_u$, $M_v$ is not solved, $M_u$ is solved and $W(M_u)\cap W(M_v)\neq\emptyset$, then $M_v$ becomes solved after the (uniquely) choose of the color of $M_u$. Moreover, if for all $v\in W(D_{i,j})$ $M_v$ is not solved, then for all $\gamma\in\Blank$ and $\gamma'\in D_{i,j}\setminus\Blank$ it holds $\gamma\not\leftrightarrow\gamma'$ because of definition of $M_i(v)$'s. Thus the 2-coloring of $D_{i,j}$ required in Step~\ref{al:antipodal}	 (if it exists) it is not unique, and thus we can choose $v\in W(D_{i,j})$ and the color of $M_v$ arbitrarily.

In this way we visit in constant time every $\gamma\in D_{i,j}$ at most $O(|W(\gamma)|)$ times, and we assign its color in $O(|W(\gamma)|)$ time. Finally, we can update $\ell_i(v)$'s and $\ell_j(v)$'s every time we visit an element $\gamma$ in $\Colored$ or we set the color of $\gamma'$ in $\Blank$ in $O(|W(\gamma)|)$ and $O(|W(\gamma')|)$ time, respectively. We update $M_v$'s in $O(1)$ time if we represent $M_v$ as a vector; we build these vectors varying $v\in W(D_{i,j})$ in  $O(\sum_{\gamma\in D_{i,j}}|W(\gamma)|)$ time. The claim follows.
\end{proof}

\subsubsection{Step~\ref{al:ct}}

Let $C$ be the clique separator of $G$ fixed at Step~\ref{al:Q}. In~\cite{mew} (proof of Proposition 9) it is shown how to build a clique path tree on the cliques in $\Gamma_C$ starting from a coloring satisfying~\ref{com:mw_i} and~\ref{com:mw_ii} in $O(|\Gamma_C|)$ time. Finally, by Theorem~\ref{th:correctness}, $f_9$ satisfies~\ref{com:mw_i} and~\ref{com:mw_ii}, and Lemma~\ref{lemma:tilde} implies that Step~\ref{al:ct} has $O(|\Gamma_C|)$ time complexity.

\subsubsection{Time complexity}\label{costo}

In Theorem~\ref{th:main_PG} we show that the algorithm \PAPERINO has $O(p(m+n))$ time complexity by summarizing the results of previous subsections and by using the following lemma proved in~\cite{schaffer}. 

\begin{lemma}[\cite{schaffer}]\label{lemma:m+n}
For every clique separator $C$ of $G$ it holds $\sum_{\gamma\in\Gamma_C} (|W(\gamma)|)\leq m+n.$
\end{lemma}

\begin{theorem}\label{th:main_PG}
Algorithm \PAPERINO can establish whether a graph $G$ is a path graph in $O(p(m+n))$ time, where $p$ is the number of cliques in $G$.
\end{theorem}
\begin{proof}
By Theorem~\ref{th:correctness}, it suffices to prove that the algorithm \PAPERINO can be executed in $O(p(m+n))$ time. Step 1 and Step 2 have $O(m+n)$ time complexity, while Step~\ref{al:rec} has $p$ times the complexity of steps~\ref{al:inizialization}, \ref{al:crossing}, \ref{al:antipodal} and \ref{al:ct}. Let $C$ be the clique separator in a recursive call of Step~\ref{al:rec}. Steps \ref{al:inizialization}, \ref{al:crossing} and \ref{al:antipodal} can be executed in $O(\sum_{\gamma\in\Gamma_C}|W(\gamma)|)$ time and Step~\ref{al:ct} has $O(|\Gamma_C|)$ time complexity. By Lemma~\ref{lemma:m+n} and being $|\Gamma_C|\leq n$ for every clique separator $C$, the claim follows.
\end{proof}

\section{Recognition algorithm for directed path graphs}\label{section:algorithm_DPG}

In this section we present algorithm \PAPERINObis that is able to recognize directed path graphs. It is based on Theorem~\ref{mew2} and on algorithm \PAPERINO, both algorithms have the same time complexity.

Thanks to Theorem~\ref{mew2}, given a clique separator $C$, we have to check whether the antipodality graph $A_C$ is 2-colorable. If so, then there are no antipodal triangle in $\Gamma_C$. Moreover, the absence of an antipodal triangle implies the absence of full antipodal triangle. Thus, if there are no full antipodal triangles in $\Gamma_C$, then the rigid structure described in Lemma~\ref{lemma:elenco_trivial_2} still holds. Consequently, we obtain the following characterization by noting that the coloring in Theorem~\ref{th:coloring} is a proper vertex coloring of $A_C$, thus it suffices to reduce the color set from $[r+1]$ to $\{0,1\}$ and modify accordingly all the conditions.


\begin{theorem}\label{th:coloring_DPG}
A chordal graph $G$ is a directed path graph if and only if $G$ is an atom or for a clique separator $C$ each graph $\gamma\in \Gamma_C$ is a directed path graph, $\U_C=(u_1,u_2,\ldots,u_r)$ has no antipodal triangle and there exists $f:\Gamma_C\rightarrow \{0,1\}$ such that:
\renewcommand{\theenumi}{\thetheorem.\arabic{enumi}}
\begin{enumerate}
\item\label{item:2-color_no_triangle} for all $u,u'\in \U_C$, if $u\leftrightarrow u'$, then $f(u)\neq f(u')$, 
\item for all $i\in[r]$, for all $\gamma\in D_i$ if $\exists u\in \U_C$ such that $\gamma\leftrightarrow u$, then $f(\gamma)=f(u_i)$,
\item for all $i<j\in[r]$, for all $\gamma\in D_{i,j}$ such that $\exists\gamma'\in D_{k}$, for $k\in\{i,j\}$, satisfying $\gamma\leftrightarrow\gamma'$, then $f(\gamma)=\{0,1\}\setminus f(u_k)$,
\item for all $D\in\mathcal{D}^C$, for all $\gamma,\gamma\in D$ such that $\gamma\leftrightarrow\gamma'$, $f(\gamma)\neq f(\gamma')$.
\end{enumerate}
\end{theorem}

Theorem~\ref{th:coloring_DPG} implies algorithm \PAPERINObis. For a proof of its correctness, we redirect the reader to Theorem~\ref{th:correctness}'s proof, by using Theorem~\ref{th:coloring_DPG} in place of Theorem~\ref{th:coloring}.

\mybreak
\textbf{\PAPERINObis}\\
\textbf{Input:} a graph $G$\\
\textbf{Output:} if $G$ is a directed path graph, then return a directed clique path tree of $G$; else QUIT
\renewcommand{\theenumi}{\thetheorem.\arabic{enumi}}
\begin{enumerate}[label=\arabic*.D, ref=\arabic*.D]\itemsep0em
\item Test if $G$ is chordal. If not, then QUIT.

\item If $G$ has at most two cliques, then return a directed clique path tree of $G$. Else find a clique separator $C$, let $\Gamma_C=\{\gamma_1,\ldots,\gamma_s\}$, $\gamma_i=G[V_i\cup C]$, be the set of connected components separated by $C$.

\item\label{step:2} Recursively test the graphs $\gamma\in \Gamma_C$. If any one is not a directed  path graph, then QUIT, otherwise, return a directed  clique path tree $T_\gamma$ for each $\gamma\in\Gamma_C$. 

\item\label{step:initialization_D} Compute $\Gamma_C/\!\!\sim$ and initialize $f(\gamma)=$ \NULL, for all $\gamma\in\Gamma_C/\!\!\sim$. Function $f$ has values in $\{0,1\}$. Compute $\U_C$ and fix an order of its element, i.e., $\U_C=(u_1,\ldots,u_r)$. Assign values of $f$ to elements in $\U_C$ so that antipodal elements have different color. If it is not possible, then QUIT. Compute $D_i$, for all $i\in[r]$, and $D_{i,j}$, for $i<j\in[r]$.

\item For all $i\in[r]$, if there exist $\gamma\in D_i$ and $u\in \U_C$ such that $\gamma\leftrightarrow u$, then $f(\gamma)\neq f(u)$.

\item For all $i\in[r]$, extend $f$ to all elements in $D_i$ so that $f(\gamma)\neq f(\gamma')$ for all $\gamma,\gamma'\in D_i$ satisfying $\gamma\not\leftrightarrow\gamma'$. If it is not possible, then QUIT.

\item For all $i<j\in[r]$,
	\begin{itemize}\itemsep0em
	\item if there exist $\gamma\in D_{i,j}$, $\gamma'\in D_i$, $\gamma''\in D_j$, such that $\gamma\leftrightarrow\gamma'$ and $\gamma\leftrightarrow\gamma''$, then QUIT,
	\item if there exist $\gamma\in D_{i,j}$ and $\gamma'\in D_j$ such that $\gamma\leftrightarrow\gamma'$, then $f(\gamma)=f(u_i)$,
	\item if there exist $\gamma\in D_{i,j}$ and $\gamma'\in D_i$ such that $\gamma\leftrightarrow\gamma'$, then $f(\gamma)=f(u_j)$.
	\end{itemize}

\item For all $i<j\in[r]$, extend $f$ to all elements in $D_{i,j}$ so that $f(\gamma)\neq f(\gamma')$ for all $\gamma,\gamma'\in D_{i,j}$ satisfying $\gamma\not\leftrightarrow\gamma'$. If it is not possible, then QUIT.

\item\label{step:11} Convert the coloring $f:\Gamma_C/\!\!\sim\rightarrow\{0,1\}$ in a directed  clique path tree of $\Gamma_C$.
\end{enumerate}

Any step of algorithm \PAPERINObis has the same time complexity of the corresponding step of algorithm \PAPERINO, also implementation details are similar. 
Consequently, the following theorem applies.

\begin{theorem}\label{th:main_DPG}
Algorithm \PAPERINObis can establish whether a graph $G$ is a directed path graph in $O(p(m+n))$ time, where $p$ is the number of cliques in $G$.
\end{theorem}

\section{Conclusions}

We presented the first recognition algorithm for both path graphs and directed path graphs. Both graph classes are characterized very similarly in~\cite{mew}, and we extended the simpler characterization of path graphs in~\cite{ab} to include directed path graphs as well; this result can be of interest itself. Thus, now these two graph classes can be recognized in the same way both theoretically and algorithmically.

On the side of path graphs, we believe that, compared to algorithms in~\cite{chaplick,schaffer}, our algorithm is simpler for several reasons: the overall treatment is shorter, the algorithm does not require complex data structures, its correctness is a consequence of the characterization in~\cite{ab}, and there are fewer implementation details to achieve the same computational complexity as in~\cite{chaplick,schaffer}.

On the side of directed path graphs, prior to this paper, it was necessary to implement two algorithms to recognize them: a recognition algorithm for path graphs as in~\cite{chaplick,schaffer}, and the algorithm in~\cite{chaplick-gutierrez} that in linear time is able to determining whether a path graph is also a directed path graph. Our algorithm directly recognizes directed path graphs in the same time complexity, and the simplification is clear.

\section*{Acknowledgments}
This work is partly funded by the HORIZON Research and Innovation Action 101135576 INTEND
``Intent-based data operation in the computing continuum''.


The author sincerely thanks the anonymous reviewers for their insightful comments and constructive feedback, which have significantly contributed to improving the quality and clarity of this manuscript.



\bibliographystyle{abbrvnat}
\bibliography{biblio.bib}

\begin{thebibliography}{23}
\providecommand{\natexlab}[1]{#1}
\providecommand{\url}[1]{\texttt{#1}}
\expandafter\ifx\csname urlstyle\endcsname\relax
  \providecommand{\doi}[1]{doi: #1}\else
  \providecommand{\doi}{doi: \begingroup \urlstyle{rm}\Url}\fi

\bibitem[Apollonio and Balzotti(2023)]{ab}
N.~Apollonio and L.~Balzotti.
\newblock Two new characterizations of path graphs.
\newblock \emph{Discrete Mathematics}, 346\penalty0 (12):\penalty0 113596,
  2023.
\newblock ISSN 0012-365X.
\newblock \doi{https://doi.org/10.1016/j.disc.2023.113596}.
\newblock URL
  \url{https://www.sciencedirect.com/science/article/pii/S0012365X23002820}.

\bibitem[Booth and Lueker(1976)]{booth-lueker}
K.~S. Booth and G.~S. Lueker.
\newblock Testing for the {C}onsecutive {O}nes {P}roperty, {I}nterval {G}raphs,
  and {G}raph {P}lanarity {U}sing {P}{Q}-{T}ree {A}lgorithms.
\newblock \emph{Journal of Computer and System Sciences}, 13\penalty0
  (3):\penalty0 335--379, 1976.
\newblock URL \url{https://doi.org/10.1016/S0022-0000(76)80045-1}.

\bibitem[Chaplick(2008)]{chaplick}
S.~Chaplick.
\newblock \emph{P{Q}{R}-{T}rees and {U}ndirected {P}ath {G}raphs}.
\newblock PhD thesis, University of Toronto, 2008.
\newblock URL \url{http://hdl.handle.net/1807/118597}.

\bibitem[Chaplick et~al.(2010)Chaplick, Gutierrez, L{\'{e}}v{\^{e}}que, and
  Tondato]{chaplick-gutierrez}
S.~Chaplick, M.~Gutierrez, B.~L{\'{e}}v{\^{e}}que, and S.~B. Tondato.
\newblock From {P}ath {G}raphs to {D}irected {P}ath {G}raphs.
\newblock In \emph{Graph Theoretic Concepts in Computer Science - 36th
  International Workshop, {WG} 2010}, volume 6410 of \emph{Lecture Notes in
  Computer Science}, pages 256--265, 2010.
\newblock URL \url{https://doi.org/10.1007/978-3-642-16926-7\_24}.

\bibitem[Corneil et~al.(1998)Corneil, Olariu, and Stewart]{corneil-olariu}
D.~G. Corneil, S.~Olariu, and L.~Stewart.
\newblock The {U}ltimate {I}nterval {G}raph {R}ecognition {A}lgorithm?
\newblock In \emph{Proceedings of the Ninth Annual {ACM-SIAM} Symposium on
  Discrete Algorithms}, pages 175--180. {ACM/SIAM}, 1998.
\newblock URL \url{http://dl.acm.org/citation.cfm?id=314613.314697}.

\bibitem[Dahlhaus and Bailey(1996)]{dah}
E.~Dahlhaus and G.~Bailey.
\newblock Recognition of {P}ath {G}raphs in {L}inear {T}ime.
\newblock In \emph{5th Italian Conference on Theoretical Computer Science
  (Revello, 1995)}, pages 201--210. World Scientific, 1996.

\bibitem[Dietz(1984)]{dietz}
P.~F. Dietz.
\newblock Intersection {G}raph {A}lgorithms.
\newblock Technical report, Cornell University, 1984.

\bibitem[Gavril(1974)]{gavril1}
F.~Gavril.
\newblock The {I}ntersection {G}raphs of {S}ubtrees in {T}rees are {E}xactly
  the {C}hordal {G}raphs.
\newblock \emph{Journal of Combinatorial Theory, Series B}, 16\penalty0
  (1):\penalty0 47--56, 1974.

\bibitem[Gavril(1975)]{gavril_DV_algorithm}
F.~Gavril.
\newblock A {R}ecognition {A}lgorithm for the {I}ntersection {G}raphs of
  {D}irected {P}aths in {D}irected {T}rees.
\newblock \emph{Discrete Mathematics}, 13\penalty0 (3):\penalty0 237--249,
  1975.
\newblock URL \url{https://doi.org/10.1016/0012-365X(75)90021-7}.

\bibitem[Gavril(1978)]{gavril_UV_algorithm}
F.~Gavril.
\newblock A {R}ecognition {A}lgorithm for the {I}ntersection {G}raphs of
  {P}aths in {T}rees.
\newblock \emph{Discrete Mathematics}, 23\penalty0 (3):\penalty0 211--227,
  1978.
\newblock URL \url{https://doi.org/10.1016/0012-365X(78)90003-1}.

\bibitem[Habib et~al.(2000)Habib, McConnell, Paul, and
  Viennot]{habib-mcconnell}
M.~Habib, R.~M. McConnell, C.~Paul, and L.~Viennot.
\newblock Lex-bfs and partition refinement, with applications to transitive
  orientation, interval graph recognition and consecutive ones testing.
\newblock \emph{Theoretical Computer Science}, 234\penalty0 (1-2):\penalty0
  59--84, 2000.
\newblock URL \url{https://doi.org/10.1016/S0304-3975(97)00241-7}.

\bibitem[Hsu(1992)]{hsu}
W.~Hsu.
\newblock A {S}imple {T}est for {I}nterval {G}raphs.
\newblock In \emph{Graph-Theoretic Concepts in Computer Science, 18th
  International Workshop, Proceedings}, volume 657 of \emph{Lecture Notes in
  Computer Science}, pages 11--16. Springer, 1992.
\newblock URL \url{https://doi.org/10.1007/3-540-56402-0\_31}.

\bibitem[Hsu and McConnell(2003)]{hsu-mcconnell}
W.~Hsu and R.~M. McConnell.
\newblock {PC} trees and circular-ones arrangements.
\newblock \emph{Theoretical Computer Science}, 296\penalty0 (1):\penalty0
  99--116, 2003.
\newblock URL \url{https://doi.org/10.1016/S0304-3975(02)00435-8}.

\bibitem[Korte and M{\"{o}}hring(1989)]{korte-mohring}
N.~Korte and R.~H. M{\"{o}}hring.
\newblock An {I}ncremental {L}inear-{T}ime {A}lgorithm for {R}ecognizing
  {I}nterval {G}raphs.
\newblock \emph{{S}{I}{A}{M} Journal on Computing}, 18\penalty0 (1):\penalty0
  68--81, 1989.
\newblock URL \url{https://doi.org/10.1137/0218005}.

\bibitem[Lekkeikerker and Boland(1962)]{lekkerkerker-boland}
C.~Lekkeikerker and J.~Boland.
\newblock Representation of a finite graph by a set of intervals on the real
  line.
\newblock \emph{Fundamenta Mathematicae}, 51\penalty0 (1):\penalty0 45--64,
  1962.

\bibitem[L{\'e}v{\^e}que et~al.(2009)L{\'e}v{\^e}que, Maffray, and
  Preissmann]{bfm}
B.~L{\'e}v{\^e}que, F.~Maffray, and M.~Preissmann.
\newblock Characterizing {P}ath {G}raphs by {F}orbidden {I}nduced {S}ubgraphs.
\newblock \emph{Journal of Graph Theory}, 62\penalty0 (4):\penalty0 369--384,
  2009.

\bibitem[McConnell and Spinrad(1999)]{mcconnel-spinrad}
R.~M. McConnell and J.~P. Spinrad.
\newblock Modular decomposition and transitive orientation.
\newblock \emph{Discrete Mathematics}, 201\penalty0 (1-3):\penalty0 189--241,
  1999.
\newblock URL \url{https://doi.org/10.1016/S0012-365X(98)00319-7}.

\bibitem[Monma and Wei(1986)]{mew}
C.~L. Monma and V.~K. Wei.
\newblock Intersection {G}raphs of {P}aths in a {T}ree.
\newblock \emph{Journal of Combinatorial Theory, Series B}, 41\penalty0
  (2):\penalty0 141--181, 1986.

\bibitem[Novick(1990)]{sc14}
M.~B. Novick.
\newblock Parallel {A}lgorithms for {I}ntersection {G}raphs.
\newblock Technical report, Cornell University, 1990.

\bibitem[Renz(1970)]{renz}
P.~Renz.
\newblock Intersection {R}epresentations of {G}raphs by {A}rcs.
\newblock \emph{Pacific Journal of Mathematics}, 34\penalty0 (2):\penalty0
  501--510, 1970.

\bibitem[Rose et~al.(1976)Rose, Tarjan, and Lueker]{sc18}
D.~J. Rose, R.~E. Tarjan, and G.~S. Lueker.
\newblock Algorithmic {A}spects of {V}ertex {E}limination on {G}raphs.
\newblock \emph{SIAM Journal on Computing}, 5\penalty0 (2):\penalty0 266--283,
  1976.

\bibitem[Sch\"{a}ffer(1993)]{schaffer}
A.~Sch\"{a}ffer.
\newblock A {F}aster {A}lgorithm to {R}ecognize {U}ndirected {P}ath {G}raphs.
\newblock \emph{Discrete Appl. Math.}, 43\penalty0 (3):\penalty0 261--295,
  1993.
\newblock ISSN 0166-218X.
\newblock URL \url{https://doi.org/10.1016/0166-218X(93)90116-6}.

\bibitem[Tarjan and Yannakakis(1984)]{sc19}
R.~E. Tarjan and M.~Yannakakis.
\newblock Simple {L}inear-{T}ime {A}lgorithms to {T}est {C}hordality of
  {G}raphs, {T}est {A}cyclicity of {H}ypergraphs, and {S}electively {R}educe
  {A}cyclic {H}ypergraphs.
\newblock \emph{SIAM Journal on computing}, 13\penalty0 (3):\penalty0 566--579,
  1984.

\end{thebibliography}
\label{sec:biblio}

\end{document}